\documentclass[11pt]{article}

\usepackage{amsfonts}
\usepackage{amssymb}
\usepackage{amsthm}
\usepackage{amsmath}
\usepackage{cite}

\newtheorem{Thm}{Theorem}
\newtheorem{Lem}[Thm]{Lemma}
\newtheorem{Cor}[Thm]{Corollary}

\theoremstyle{definition}

\newtheorem{Rem}{Remark}

\DeclareMathOperator{\lmod}{mod}

\begin{document}

\title{On the distinctness of binary sequences derived from $2$-adic expansion of m-sequences over finite prime fields}
\author{Yupeng Jiang and  Dongdai Lin\\\\ \quad State Key Laboratory Of Information Security,\\
Institute of Information Engineering,\\
Chinese Academy of Sciences, Beijing 100093, P.R. China\\ E-mail: \{jiangyupeng,ddlin\}@iie.ac.cn}
\date{}
\maketitle

\begin{abstract}
Let $p$ be an odd prime with $2$-adic expansion $\sum_{i=0}^kp_i\cdot2^i$. For a sequence
$\underline{a}=(a(t))_{t\ge 0}$ over $\mathbb{F}_{p}$, each $a(t)$ belongs to $\{0,1,\ldots, p-1\}$ and
has a unique $2$-adic expansion $$a(t)=a_0(t)+a_1(t)\cdot 2+\cdots+a_{k}(t)\cdot2^k,$$ with $a_i(t)\in\{0, 1\}$.
Let $\underline{a_i}$ denote the binary sequence $(a_i(t))_{t\ge 0}$ for $0\le i\le k$. Assume $i_0$ is the
smallest index $i$ such that $p_{i}=0$ and $\underline{a}$ and $\underline{b}$ are two different m-sequences generated by
a same primitive characteristic polynomial over
$\mathbb{F}_p$. We prove that for $i\neq i_0$ and $0\le i\le k$, $\underline{a_i}=\underline{b_i}$ if and only if
$\underline{a}=\underline{b}$, and for $i=i_0$, $\underline{a_{i_0}}=\underline{b_{i_0}}$ if and only if
$\underline{a}=\underline{b}$ or $\underline{a}=-\underline{b}$. Then the period of $\underline{a_i}$ is equal
to the period of $\underline{a}$ if $i\ne i_0$ and half of the period of $\underline{a}$ if $i=i_0$. We also
discuss a possible application of the binary sequences $\underline{a_i}$.
\end{abstract}



\section{$2$-adic expansion of sequences}

Let $p$ be an odd prime and $\mathbb{F}_{p}$ be the finite field of order $p$. We
identify this field with the set $\{0,1, \ldots, p-1\}$. Assume $p$ has the unique $2$-adic expansion
$$p=p_0+p_1\cdot 2+\cdots+p_k\cdot 2^k$$
with $p_i\in \{0,1\}$ and $p_0=p_k=1$. Each element in the field $\mathbb{F}_{p}$ also has a unique $2$-adic
expansion. For a sequence $\underline{a}=(a(t))_{t\ge 0}$ over $\mathbb{F}_{p}$, we have $0\le a(t)\le p-1$, and then
$$a(t)=a_{0}(t)+a_1(t)\cdot 2+\cdots+a_{k}(t)\cdot 2^k$$ with $a_{i}(t)\in \{0,1\}$. The binary sequence
$\underline{a_i}=(a_i(t))_{t\ge 0}$ is called the $i$th level sequence of $\underline{a}$, and
$$\underline{a}=\underline{a_0}+\underline{a_1}\cdot 2+\cdots+\underline{a_k}\cdot 2^k$$
is called the $2$-adic expansion of the sequence $\underline{a}$.

Let $\underline{a}$ and $\underline{b}$ be two difference m-sequences generated by a same primitive polynomial
over $\mathbb{F}_{p}$. For more details about linear recurring sequences, see \cite{LiNi}.
It is natural to ask whether or not we have $\underline{a_i}=\underline{b_i}$ for $i$ satisfying $0\le i \le k$.
In \cite{ZhuQ2008}, Zhu and Qi proved that $\underline{a_0}=\underline{b_0}$
if and only if $\underline{a}=\underline{b}$. If $p$ is not a Mersenne prime, then there is an $i$ such that $p_i=0$.
Denote the smallest $i$ with $p_i=0$ by $i_0$.
In \cite{Zheng}, Zheng proved that $\underline{a_{i_0}}=\underline{b_{i_0}}$ if
$\underline{a}=-\underline{b}$. For general $i$, there is no result in the literature.

In this article, we will prove that for $i\ne i_0$, we have $\underline{a_i}=\underline{b_i}$ if and only
if $\underline{a}=\underline{b}$, and for $i=i_0$, we have $\underline{a_{i_0}}=\underline{b_{i_0}}$ if and only
if $\underline{a}=\underline{b}$ or $\underline{a}=-\underline{b}$. We devote Section $2$ to prove our main results.
In Section $3$, we discuss some possible applications.

\section{Main results}
In this section, we will prove our main results. First, we give some notations. Let $p$, $k$, $i_0$
be defined as above. As in \cite{ZhQT}, the notation $[n]_{\lmod m}$ means the least nonnegative integer residue
of $n$ modulo $m$. For an element $a\in \{0, 1, \ldots, p-1\}$, there is a unique $2$-adic expansion
$$a=a_0+a_1\cdot 2+\cdots+a_k\cdot 2^k.$$ Then we can define a family of maps $B_i$ for $0\le i \le k$ as follows:
\begin{align*}
B_i:\quad \mathbb{F}_p&\rightarrow\{0,1\}\\
a&\mapsto B_i(a)=a_i.
\end{align*}
Then $B_i(a)=0$ means that $0\le [a]_{\lmod 2^{i+1}}<2^i$ and $B_i(a)=1$ means that $2^{i}\le [a]_{\lmod 2^{i+1}}<2^{i+1}$.
The $2$-adic $i$th level sequence $\underline{a_{i}}$ is just $(B_{i}(a(t)))_{t\ge 0}$. For the maps $B_{i}$'s,
we have the following two lemmas. The notation
$\lambda\cdot \underline{a}$ means the sequence $([\lambda\cdot a(t)]_{\lmod p})_{t\ge 0}$.

\begin{Lem}\label{Lem:p-1}
Let $0\le i\le k$. Then $B_i(a)=B_i([-a]_{\lmod p})$ holds for all $a\in\mathbb{F}_{p}$ if and only if $i=i_0$.
\end{Lem}
\begin{proof}
``$\Rightarrow$".Assume $B_i(a)=B_i([-a]_{\lmod p})$ holds for all $a\in \mathbb{F}_{p}$. Let
$$a=x\cdot 2^i+y\qquad p=z\cdot 2^i+w$$ with $0\le y,w<2^i$. Then we have $B_i(a)=[x]_{\lmod 2}$ and
$B_{i}(p)=[z]_{\lmod 2}$. If $B_{i}(p)=1$, then we can choose $a\neq 0$ and $0\le y\le w$, then
$$B_{i}([-a]_{\lmod p})=B_{i}(p-a)=[z-x]_{\lmod 2}\neq [x]_{\lmod 2}=B_{i}(a).$$
If $B_{i}(p)=0$ and $0\le w<2^i-1$, we choose $a=w+1$. Then we have $p-a=(z-1)\cdot 2^i+2^i-1$.
Thus $$B_{i}([-a]_{\lmod p})=B_{i}(p-a)=[z-1]_{\lmod 2}=1\neq 0=B_{i}(a).$$
Thus for $B_{i}(p)=0$, we must have $w=2^i-1$, which just means $i=i_0$.

``$\Leftarrow$". Assume $i=i_0$. A proof can be seen in \cite[Theorem 5.3]{Zheng}.
For completeness, we give a proof here. Since $a=0$ is a trivial case, we assume $0<a<p$. Let $a=x\cdot 2^{i_0}+y$ with
$0\le y\le 2^{i_0}-1$. According to the definition of $i_0$, we have $p=z\cdot 2^{i_0}+w$ with even $z$ and $w=2^{i_0}-1$.
Thus $[-a]_{\lmod p}=p-a=(z-x)\cdot 2^{i_0}+w-y$. We have
$$B_{i_0}(p-a)=[z-x]_{\lmod 2}=[x]_{\lmod 2}=B_{i_0}(a).$$
The proof is complete.
\end{proof}

For $1<\lambda<p-1$, we have the following result.

\begin{Lem}\label{Lem:lambda}
Let $1<\lambda<p-1$ and $0\le i\le k$. Then $B_i(a)=B_i([\lambda\cdot a]_{\lmod p})$ can not hold for
all $a\in\mathbb{F}_{p}$.
\end{Lem}

\begin{proof}
We first prove for the case $i=0$. Then $B_{0}(a)=0$ if and only if $a\equiv0\mod 2$.
A proof can also be seen in \cite[Lemma 4.6]{ZhuQ2008}.
Assume $B_0(1)=B_0(\lambda)=1$, then $\lambda\equiv 1\mod 2$. If not, the equality does not hold for $a=1$.
Since $\lambda>1$, there exists some $a\in\{1,2,\ldots, p-1\}$ such that
$p<\lambda\cdot a<2p$. Then $[\lambda\cdot a]_{\lmod p}=\lambda\cdot a-p$, and we have
$$\lambda\cdot a-p\equiv a-1\not\equiv a\mod 2.$$ For this $a$, $B_0(a)\neq B_0([\lambda\cdot a]_{\lmod p})$.

Now we are going to proof for the case $1\le i\le k$. We assume $B_i(a)=B_i([\lambda\cdot a]_{\lmod p})$
holds for all $a\in \mathbb{F}_{p}$. Let
$$\lambda=x\cdot 2^{i+1}+y\qquad p=z\cdot 2^{i+1}+w$$
with $0\le y, w< 2^{i+1}$. Since $B_{i}(\lambda)=B_{i}(1)=0$, we have $0\le y<2^i$. Also
we have $B_{i}(p-1)=B_{i}([\lambda\cdot (p-1)]_{\lmod p})=B_{i}(p-\lambda)$. We divide the proof
into two cases: $B_{i}(p)=0$ and $B_i(p)=1$.

If $B_{i}(p)=0$, then $0\le w<2^i$ and from $p>2^k\ge2^i>w$, we have $z>0$.
Since $p$ is odd and so is $w$, we have $1\le w<2^i$ and
then $B_{i}(p-1)=0$. So $B_{i}(p-\lambda)=0$. Then
$$0\le y <2^i\qquad 0\le [w-y]_{\lmod 2^{i+1}}<2^i.$$
Thus $0\le y\le w$. Now we
count the number of $b$'s such that $B_i(b)=0$ and $B_{i}([b+\lambda]_{\lmod p})=1$.
For each $b$, there is a unique $a$ such that $[\lambda\cdot a]_{\lmod p}=b$. Then $B_{i}(a)=B_{i}(b)=0$
and $B_{i}(a+1)=B_{i}([b+\lambda]_{\lmod p})=1$,
which means $[a]_{\lmod 2^{i+1}}=2^i-1$. There are totally $z$ such $a$'s in the interval $0\le a<p$.
Denote $[b]_{\lmod 2^{i+1}}$
by $t$. If $0\le b<p-\lambda$, then $[b+\lambda]_{\lmod p}=b+\lambda$ and we have
$$0\le t<2^i\qquad 2^i\le [t+y]_{\lmod 2^{i+1}}=t+y< 2^{i+1}, $$
which means $2^i-y\le t< 2^i$. As $p-\lambda=(z-x)\cdot 2^{i+1}+(w-y)$ and $w-y<2^i-y$, there are
$(z-x)y$ such $b$'s in this interval. If $p-\lambda\le b<p$, then $[b+\lambda]_{\lmod p}=b+\lambda-p$ and we have
$$0\le t< 2^i\qquad 2^i\le [t+y-w]_{\lmod 2^{i+1}}<2^{i+1}.$$
Since $-2^i<t+y-w\le t<2^i$, then we must have $t+y-w<0$ and $0\le t<w-y$. There are $x(w-y)$ such $b$'s in this interval.
Then we have $$z=(z-x)y+x(w-y).$$
If one of the four nonnegative integers $z-x$, $y$, $x$, $w-y$ is zero, we can reduce to
$\lambda=1$ or $p-1$, which contradicts to the condition that $1<\lambda<p-1$.\\
(1) If $y=0$, then $z=xw$. Thus $w\mid p$ and $w<p$, so $w=1$ and $z=x$, which means $\lambda=p-1$.\\
(2) If $w-y=0$, then $z=(z-x)w$. For the same reason we have $w=1$. Then $x=0$ and $y=w=1$. Thus $\lambda=1$.\\
(3) If $x=0$, then $z=zy>0$. So $y=1$ and $\lambda=1$.\\
(4) If $z-x=0$, then $z=z(w-y)>0$. So $w-y=1$ and $\lambda=p-1$.\\
Now we assume all of the four integers are positive. Then
$$z=(z-x)y+x(w-y)\ge (z-x)+x=z.$$
So we have $y=w-y=1$ and $w=2$, which is impossible since $w$ is odd.

If $B_{i}(p)=1$, then $2^i\le w<2^{i+1}$. Since $w$ is odd, then $w>2^i$ and $B_{i}(p-1)=1$.
So $B_{i}(p-\lambda)=1$. Then
$$0\le y<2^i \qquad 2^i\le [w-y]_{\lmod 2^{i+1}}=w-y<2^{i+1}.$$
Thus $0\le y\le w-2^i$.
Again  we count the number of $b$'s such that $B_i(b)=0$ and $B_{i}([b+\lambda]_{\lmod p})=1$.
For each $b$, there is a unique $a$ such that $[\lambda\cdot a]_{\lmod p}=b$, then $B_{i}(a)=0$ and $B_{i}(a+1)=1$,
which means $[a]_{\lmod 2^{i+1}}=2^i-1$. There are totally $z+1$ such $a$'s in the interval $0\le a<p$.
Denote $[b]_{\lmod 2^{i+1}}$
by $t$. If $0\le b<p-\lambda$, then we have
$$0\le t<2^i\qquad 2^i\le [t+y]_{\lmod 2^{i+1}}=t+y< 2^{i+1}, $$
which means $2^i-y\le t< 2^i$. As $p-\lambda=(z-x)\cdot 2^{i+1}+(w-y)$ and $w-y\ge 2^i$, there are
$(z-x+1)y$ such $b$'s in this interval. If $p-\lambda\le b<p$, then $[b+\lambda]_{\lmod p}=b+\lambda-p$ and we have
$$0\le t< 2^i\qquad 2^i\le [t+y-w]_{\lmod 2^{i+1}}<2^{i+1}.$$
Since $-2^{i+1}<t+y-w\le t-2^i<0$, we must have $-2^i\le t+y-w<0$ and then $w-y-2^i\le t<2^i$.
There are $x(2^{i+1}-w+y)$ such $b$'s in this interval. Then we have
$$z+1=(z-x+1)y+x(2^{i+1}-w+y)=(z+1)y+x(2^{i+1}-w).$$
If $y>0$, as $x(2^{i+1}-w)\ge 0$, we have $y=1$ and $x=0$. Then $\lambda=1$, which contradicts to
the condition $1<\lambda<p-1$. If $y=0$, then $z+1=x(2^{i+1}-w)$. We have
$$p=z\cdot2^{i+1}+w=z(2^{i+1}-w)+(z+1)w.$$
Then $(2^{i+1}-w)\mid p$. Since $w>2^i$, then $2^{i+1}-w<2^i<w\le p$. So $2^{i+1}-w=1$ and then $z+1=x$,
which is impossible since $z\ge x$. Then the proof is complete.
\end{proof}

Now we can prove the following theorem about level sequences of $2$-adic expansion of m-sequences.

\begin{Thm}\label{Thm:levelsequences}
Let $p$, $k$, $i_0$ be defined as above. Assume $\underline{a}=(a(t))_{t\ge 0}$ and
$\underline{b}=(b(t))_{t\ge 0}$ are two different
m-sequences generated by a same primitive polynomial over $\mathbb{F}_{p}$.
Then $\underline{a_{i_0}}=\underline{b_{i_0}}$ if and only if
$\underline{a}=\underline{b}$ or $\underline{a}=-\underline{b}$, and for $i\neq i_0$,
$\underline{a_{i}}=\underline{b_{i}}$ if and only if $\underline{a}=\underline{b}$.
\end{Thm}
\begin{proof}
First, assume $\underline{a}$ and $\underline{b}$ are linearly independent.
For those $t$ with $a(t)=0$, $b(t)$ can be any element in $\mathbb{F}_{p}$, then
$\underline{a_i}\neq\underline{b_i}$ for every $i$ satisfying $0\le i\le k$. If
$\underline{a}$ and $\underline{b}$ are linearly dependent, then there is $\lambda\in\{1,2,\ldots, p-1\}$
such that $\underline{a_i}=\lambda\cdot\underline{b_i}$. Then $\underline{a_{i}}=\underline{b_{i}}$
if and only if  $B_{i}(a)=B_{i}([\lambda\cdot a]_{\lmod p})$ holds for all $a\in\mathbb{F}_{p}$.
For $\lambda=p-1$, by Lemma \ref{Lem:p-1}, we
have $\underline{a_{i}}=\underline{b_{i}}$ if and only if $i=i_0$.
For $1<\lambda<p-1$, by Lemma \ref{Lem:lambda}, we have $\underline{a_{i}}\neq\underline{b_{i}}$ for
all $i$'s. The proof is complete.
\end{proof}

\begin{Rem}
When in the $2$-adic expansion $p=\sum_{i=0}^k p_i\cdot 2^i$, every $p_i$ is equal to $1$,
then $p$ is a Mersenne prime $2^{k+1}-1$ and no such $i_0$ exists. By the above theorem,
$\underline{a_i}=\underline{b_{i}}$ if and only if $\underline{a}=\underline{b}$.
This result can be proved by only using the case $i=0$ of Lemma \ref{Lem:lambda} and the following
fact. For $i>0$, if $1\le a=x\cdot 2^i+y<p$ with $0\le y<2^i$, then $0\le x<2^{k+1-i}$. We have
$$2^{k+1-i}\cdot a=x\cdot2^{k+1}+y\cdot 2^{k+1-i}\equiv y\cdot 2^{k+1-i}+x\mod p, \eqno(*)$$
and then $[2^{k+1-i}\cdot a]_{\lmod p}=y\cdot 2^{k+1-i}+x$. So
$$B_{i}(a)=[x]_{\lmod 2}=B_{0}([2^{k+1-i}\cdot a]_{\lmod p}).$$
The sequence $\underline{a_i}$ is equal to $\underline{c_0}$ with $\underline{c}=2^{k+1-i}\cdot \underline{a}$.
For the same reason $\underline{b_i}$ is equal to $\underline{d_0}$ with $\underline{d}=2^{k+1-i}\cdot \underline{b}$.
Thus $\underline{a_i}=\underline{b_i}$ if and only if $\underline{c_0}=\underline{d_0}$. We have
$\underline{c_0}=\underline{d_0}$ if and only if $\underline{c}=\underline{d}$, and if
and only if $\underline{a}=\underline{b}$.
\end{Rem}

We have the following corollary about the period of $\underline{a_i}$.

\begin{Cor}
Let $\underline{a}$ be an m-sequence of order $n$ over $\mathbb{F}_{p}$ and $T=p^n-1$.
Then the period of $\underline{a_{i}}$ is equal to $T$ if $i\ne i_0$ and equal to
$T/2$ if $i=i_0$.
\end{Cor}
\begin{proof}
The period of $\underline{a}$ is $T$\cite{LiNi}. Denote the period of $\underline{a_i}$
by $T_i$. Then we have $T_{i}\mid T$ and $B_{i}(a(t+T_{i}))=B_{i}(a(t))$ holds for all
$t\ge 0$. Let $\underline{b}=((b(t))_{t\ge 0}=(a(t+T_i))_{t\ge 0}$. Then we have
$\underline{b_{i}}=\underline{a_{i}}$. By Theorem \ref{Thm:levelsequences}, if $i\ne i_0$,
$\underline{a}=\underline{b}$, which means $T\mid T_i$. So $T_{i}=T$. If $i=i_0$, we have
$\underline{a}=-\underline{b}$ or $\underline{a}=\underline{b}$. The smallest $T_{i_{0}}$
such that $a(t)=[-a(t+T_{i_{0}})]_{\lmod p}$ is $T/2$\cite{LiNi}. So the period of
$\underline{a_{i_{0}}}$ is $T/2$. The proof is complete.
\end{proof}

\section{Applications}
In this section, we discuss a possible application of the $2$-adic level sequences. In \cite{Zuc}, the
ZUC algorithm adopts primitive sequences over the prime field of order $2^{31}-1$ as drive sequences,
and the $0$th level sequence is used.
The addition and multiplication modulo $2^{31}-1$ have a fast implementation. Since $2^{31}-1$ is a Mersenne
prime, we can see from ($*$) that the modulo $2^{31}-1$ multiplication by powers of $2$ can be done
by cyclic shift. For addition, if $a+b=c\cdot 2^{31}+d$ with $0\le a, b<2^{31}-1$, then
$[a+b]_{\lmod 2^{31}-1 }=c+d$.

If only a few coefficients of the $2$-adic expansion of $a$ are equal to $1$, then for $p$ of the form $2^n-a$,
the addition and multiplication modulo $p$ also have a fast implementation. We suggest that the ZUC algorithm may
use m-sequences over prime fields of order $2^{n}-a$ instead of the field of order $2^{31}-1$. There are two reasons.
First, we can choose $n$ to be $32$ and $64$ to make the operations more suitable to
$32$ and $64$ bit platforms respectively. Second, there are more primes and more level sequences to use.
If we let $a$ be of the form $2^i+1<2^{n-1}$, it is easy to check that $2^{32}-a$ is prime for $i=2, 4,6,23,24,25,29$
and $2^{64}-a$ is prime for $i=8, 10, 29$. Furthermore, these primes are not Mersenne primes and then
for different $i$ and $j$, the $i$th level sequence of an m-sequence is not necessarily equal to
the $j$th level sequence of another m-sequence. There are more level sequences can be used for each prime.

\end{document}